\documentclass[a4paper,leqno,11pt]{article}

\usepackage{amsmath,amssymb,amsthm,%showkeys,
verbatim%,refcheck
}

\usepackage[a4paper,top=3cm,bottom=3cm,left=3.2cm,right=3cm]{geometry}
\usepackage{tikz}
\tikzstyle{nodo}=[circle,draw,fill,inner sep=0pt, minimum size=0.5*width("k")]
\tikzstyle{infinito}=[circle,inner sep=0pt,minimum size=0mm]

\newcommand\R{{\mathbb R}}

\newcommand\f{\frac}
\newcommand\dx{{\,dx}}

\newcommand\muR{\mu_\R}
\newcommand\muG{\mu_\G}
\newcommand\muRp{\mu_{\R^+}}

\newcommand\elevel{{\mathcal E}}

\newcommand\G{\mathcal G}
\newcommand\B{\mathcal B}
\newcommand\K{\mathcal K}

\newcommand\HH{\mathcal H}
\newcommand\be{\begin{equation}}
\newcommand\ee{\end{equation}}

\newcommand\eps{\varepsilon}

\newtheorem{theorem}{Theorem}[section]
\newtheorem{proposition}[theorem]{Proposition}
\newtheorem{lemma}[theorem]{Lemma}
\newtheorem{corollary}[theorem]{Corollary}

\theoremstyle{remark}
\newtheorem{remark}[theorem]{Remark}
\newtheorem*{remark*}{Remark}

\theoremstyle{definition}
\newtheorem{definition}[theorem]{Definition}

\date{}

\title{Negative energy ground states \\ for the $L^2$-critical NLSE  on metric graphs}

\author{Riccardo Adami\thanks{Author partially supported by the FIRB 2012 project ``Dispersive dynamics: Fourier
Analysis and Variational Methods".}, Enrico Serra, Paolo Tilli \\ \ \\{\small  Dipartimento di Scienze
Matematiche ``G.L. Lagrange'', Politecnico di Torino } \\ {\small
Corso Duca degli Abruzzi, 24, 10129 Torino, Italy}}

\begin{document}

\maketitle

\begin{abstract}
We investigate the existence of ground states with prescribed mass for the focusing nonlinear Schr\"odinger
equation  with $L^2$-critical power nonlinearity on noncompact quantum graphs.
We prove that, unlike the case of the real line, for certain classes of graphs
there exist ground states with negative energy for a whole interval of masses.
A key role is played by a thorough analysis of Gagliardo-Nirenberg inequalities and on
estimates of the optimal constants. Most of the techniques are new and suited to the investigation of
variational problems on metric graphs.
 \end{abstract}

\smallskip

\noindent{\small Keywords: Minimization, metric graphs, critical growth,
  nonlinear Schr\"odinger \\ \hbox{} \hskip 1.65cm Equation.}

\section{Introduction}
In this paper we investigate the existence of ground states for
the \emph{critical} NLS energy functional
\begin{equation}
\label{NLSe} E (u,\G)
  =  \frac 1 2 \| u' \|^2_{L^2 (\G)}
- \frac 1 6  \| u \|^6_{L^6 (\G)} =\frac 1 2 \int_\G |u'|^2dx
-\frac 1 6 \int_\G |u|^6 dx
\end{equation}
on a noncompact metric graph $\G$, under the {\em mass constraint}
\begin{equation}
\label{mass} \| u \|^2_{L^2 (\G)} \ = \ \mu.
\end{equation}
The subcritical case, where the $L^6$ norm is replaced by an $L^p$
norm with
 $p\in (2,6)$, has been investigated in \cite{ast,ast2}.
The energy in \eqref{NLSe} is \emph{critical} in the sense that, under
the
mass-preserving transformations
\[
u(x)\quad\mapsto\quad u_\lambda(x):=\lambda^{1/2} \,u(\lambda\, x)\qquad
(\lambda>0),
\]
the kinetic and the potential terms in \eqref{NLSe} scale in the same way, namely
\begin{equation}
\label{scalE}
E(u_\lambda,\lambda^{-1}\G) \quad=\quad
\lambda^2\, E(u,\G),
\end{equation}
which is typical of critical problems with a strong loss of compactness.

Throughout the paper, $\G$ denotes a \emph{noncompact metric graph},
i.e. a connected metric space obtained by gluing together,
by the identification of some of their endpoints, a finite number of closed
line intervals (not necessarily bounded), according to the topology of
a graph, self-loops and multiple edges being allowed. Any bounded edge $e$ is identified with
an interval $[0,\ell_e]$, while
unbounded edges are referred to as  ``half-lines'', and are identified with (copies of)
the positive half-line $\R^+=[0,+\infty)$; at least one edge is assumed to be unbounded,
so that $\G$ is noncompact
(two very special cases are when $\G=\R^+$
and when $\G=\R$, the latter being obtained by gluing together two
copies of $\R^+$). We refer to Section~\ref{sec2} (see also \cite{berkolaiko, exner, ast}) for more
details.

In this framework, by a ``ground state of mass $\mu$'' we mean a solution
 to the minimization problem
 \begin{equation}
\label{minprob} \min_{u\in H^1_\mu(\G)} E(u,\G),\qquad
H^1_\mu(\G):=\left\{ u\in H^1(\G) :\,\Vert
u\Vert_{L^2(\G)}^2=\mu\right\},
\end{equation}
for which it is
clearly sufficient to work with real valued, nonnegative functions.
Obviously ground states solve, for some $\omega \in \R$, the stationary quintic NLS equation
\[
u'' + |u|^4u = \omega u
\]
on each edge of $\G$, with Kirchhoff boundary conditions at the vertices (see Prop. 3.3 in \cite{ast}).

The existence of ground states for a given $\mu$ is strictly related to the behavior
of the \emph{ground-state energy level} function
\begin{equation}
\label{defelevel}
\elevel_\G (\mu)=\inf_{u\in H^1_\mu(\G)} E(u,\G),\quad\mu\geq 0,
\end{equation}
which will play a central role throughout this paper.

As is wellknown (see Sec. 2), when $\G=\R$
there exists a \emph{critical mass} $\muR$
such that the minimization problem \eqref{minprob} has a solution
if and only if $\mu=\muR$, and the same occurs when $\G=\R^+$
(with a \emph{smaller} critical mass $\muRp=\muR/2$).
This severe restriction is due to the scaling rule \eqref{scalE} and the dilation-invariance
of
$\R$ and $\R^+$.  Thus, when $\G=\R$ or $\G=\R^+$, the
minimization process \eqref{minprob} is extremely unstable and, in
a sense, of little interest.

When $\G$ is a generic (noncompact) metric graph, however,
the problem
can be highly nontrivial and, depending on the topology of $\G$,
 entirely new phenomena may
arise, such as problem \eqref{minprob} having solutions if, and only
if,  $\mu$ belongs to some  \emph{whole interval} of masses.

For each graph $\G$ we can define, in a natural way, a
\emph{critical mass} $\muG$, that depends on $\G$ via the
best constant $K_\G$ in the Gagliardo-Nirenberg inequality
\eqref{GN}, and it turns out that $\muRp\leq \muG\leq\muR$,
so that $\R^+$ and $\R$ are extremal graphs, as concerns the
critical mass (see Proposition~\ref{intermediate}). The mass $\muG$
is the precise threshold such that $\elevel_\G(\mu)<0$ (possibly $-\infty$)
as soon as $\mu>\muG$ and, on a general ground, a \emph{necessary condition}
for the existence of ground states in \eqref{minprob} is that
$\mu\in [\muG,\muR]$ (see Proposition~\ref{banali}).

This condition, however,
is far from being sufficient: the true nature of problem~\eqref{minprob}
strongly depends on the topology of $\G$, and the
following (mutually exclusive) cases are possible:
\newcommand\casoI{(a)}
\newcommand{\casoII}{(b)}
\newcommand{\casoIII}{(c)}
\newcommand{\casoIV}{(d)}
\begin{itemize}
\item[\casoI] $\G$ has a \emph{terminal point} (a tip, Fig. \ref{figbaffo}). Then $\muG=\muRp$, and
problem~\eqref{minprob} has no solution unless $\mu=\muRp$ and
$\G$ is isometric to $\R^+$.
\item[\casoII] $\G$ admits a \emph{cycle covering} (Fig. \ref{figH}). Then $\muG=\muR$, and
problem~\eqref{minprob} has no solution unless $\mu=\muR$ and
$\G$ is isometric to $\R$ (or to one of a few very special, and completely classified,
other structures, see Theorem 2.5 in \cite{ast}).
\item[\casoIII]  $\G$ has exactly one half-line and no terminal point (Fig. \ref{figonehalf}). Then $\muG=\muRp$,
and problem~\eqref{minprob} has a solution if and only if $\mu\in(\muRp,\muR]$.
\item[\casoIV] In all other cases (Fig. \ref{figD}): if $\muG<\muR$, then
problem~\eqref{minprob} has a solution if and only if $\mu\in[\muG,\muR]$.
\end{itemize}
\begin{figure}
\begin{minipage}{0.48\textwidth}
\begin{tikzpicture}[xscale= 0.5,yscale=0.5] %%% un terminal edge
\node at (-.5,2) [nodo] (02) {};
\node at (2,2) [nodo] (22) {};
\node at (2,4) [nodo] (24) {};
\node at (3.6,1.6) [nodo] (42) {};
\node at (3,3) [nodo] (33) {};
\node at (5,2) [nodo] (52) {};
\node at (3,3) [nodo] (32) {};
\node at (4,3) [nodo] (43) {};
\node at (-1,4) [nodo] (04) {};
\node at (2,4) [nodo] (24) {};
\node at (.5,1) [nodo] (11) {};
\node at (2,0) [nodo] (20) {};
\node at (4,0) [nodo] (40) {};
\node at (-4,2) [minimum size=0pt] (meno) {};
\node at (-4,4) [minimum size=0pt] (menoalt) {};
\node at (7.9,2) [minimum size=0pt] (piu) {};
\node at (-4.1,2) [infinito]  (infmeno) {$\scriptstyle\infty$};
\node at (-4.1,4) [infinito]  (infmenoalt) {$\scriptstyle\infty$};
\node at (8,2) [infinito]  (infpiu) {$\scriptstyle\infty$};

\node at (6,4) [nodo] (term){};
\draw [-] (43)--(term);

\draw[-] (02)--(04);
\draw[-] (04)--(24);
\draw[-] (04)--(22);
\draw[-] (24)--(22);
\draw[-] (02)--(11);
\draw[-] (11)--(22);
\draw[-] (11)--(20);
\draw[-] (20)--(22);
\draw[-] (22)--(33);
\draw[-] (24)--(33);
\draw[-] (24)--(43);
\draw[-] (33)--(43);
\draw[-] (43)--(52);
\draw[-] (33)--(42);
\draw[-] (20)--(42);
\draw[-] (20)--(40);
\draw[-] (40)--(42);
\draw[-] (42)--(52);
\draw[-] (40)--(52);
\draw[-] (02)--(meno);
\draw[-] (52)--(piu);
\draw[-] (04)--(menoalt);
\end{tikzpicture}
\caption{\footnotesize{{case (a)}. }}
\label{figbaffo}
\end{minipage}
\begin{minipage}{0.48\textwidth}
\begin{tikzpicture}
[scale=1,style={circle,inner sep=0pt,minimum size=7mm}] %%% ipotesi H
\node at (0,0) [nodo] (1) {};
\node at (-1.5,0) [infinito]  (2){$\scriptstyle\infty$};
\node at (1,0) [nodo] (3) {};
\node at (0,2) [nodo] (4) {};
\node at (-1.5,2) [infinito] (5) {$\scriptstyle\infty$};
\node at (2,0) [nodo] (6) {};
\node at (3,0) [nodo] (7) {};
\node at (2,2) [nodo] (8) {};
\node at (3,2) [nodo] (9) {};
\node at (4.5,0) [infinito] (10) {$\scriptstyle\infty$};

\node at (4.5,2) [infinito] (12) {$\scriptstyle\infty$};
\draw [-] (1) -- (2) ;
 \draw [-] (1) -- (3);
 \draw [-] (1) -- (4);
 \draw [-] (3) -- (4);
 \draw [-] (5) -- (4);
 \draw [-] (3) -- (6);
 \draw [-] (6) -- (7);
 \draw [-] (6) to [out=-40,in=-140] (7);
\draw [-] (3) to [out=10,in=-35] (1.4,0.7); 
\draw [-] (1.4,0.7) to [out=145,in=100] (3); 
 \draw [-] (6) to [out=40,in=140] (7);
 \draw [-] (6) -- (8);
 \draw [-] (6) to [out=130,in=-130] (8);
 \draw [-] (7) -- (8);
 \draw [-] (8) -- (9);
  \draw [-] (7) -- (9);
  \draw [-] (9) -- (12);
  \draw [-] (7) -- (10);
\end{tikzpicture}

\caption{\footnotesize{case (b). }}
\label{figH}

\end{minipage}
\vspace{1cm}

\begin{minipage}{0.48\textwidth}
\vspace{0.7cm}
\begin{tikzpicture}[xscale= 0.5,yscale=0.5]  %%% una sola semiretta
\node at (4.5,2) [nodo] (02) {};
\node at (7,2) [nodo] (22) {};
\node at (7,4) [nodo] (24) {};
\node at (8.6,1.6) [nodo] (42) {};
\node at (8,3) [nodo] (33) {};
\node at (10,2) [nodo] (52) {};
\node at (8,3) [nodo] (32) {};
\node at (9,3) [nodo] (43) {};
\node at (5.5,3) [nodo] (04) {};
\node at (7,4) [nodo] (24) {};
\node at (5.5,1) [nodo] (11) {};
\node at (7,0) [nodo] (20) {};
\node at (9,0) [nodo] (40) {};
\node at (-3,2) [minimum size=0pt] (meno) {};

\node at (-3.1,2) [infinito]  (infmeno) {$\scriptstyle\infty$};

\draw[-] (02)--(04);
\draw[-] (04)--(24);
\draw[-] (04)--(22);
\draw[-] (24)--(22);
\draw[-] (02)--(11);
\draw[-] (11)--(22);
\draw[-] (11)--(20);
\draw[-] (20)--(22);
\draw[-] (22)--(33);
\draw[-] (24)--(33);
\draw[-] (24)--(43);
\draw[-] (33)--(43);
\draw[-] (43)--(52);

\draw[-] (33)--(42);
\draw[-] (20)--(42);
\draw[-] (20)--(40);
\draw[-] (40)--(42);
\draw[-] (42)--(52);
\draw[-] (40)--(52);
\draw[-] (02)--(meno);

\end{tikzpicture}
\caption{\footnotesize{case (c). }}
\label{figonehalf}
\end{minipage}
\begin{minipage}{0.48\textwidth}
\vspace{0.7cm}
\hspace{0.5cm}
\begin{tikzpicture}[xscale= 0.5,yscale=0.5] %%%  palina
\node at (6,0) [infinito]  (1) {$\scriptstyle\infty$};
\node at (12,0) [nodo] (2) {};
\node at (18,0) [infinito]  (3) {$\scriptstyle\infty$};
\node at (16,0)  [minimum size=0pt] (5) {};
\node at (8,0) [minimum size=0pt] (4) {};
\node at (12,2.4) [nodo] (6) {};
\draw[-] (1) -- (3);
\draw(12,3) circle (0.6);
\draw[-] (2) -- (6);
\end{tikzpicture}
\caption{\footnotesize{case (d). }}
\label{figD}
\end{minipage}
\end{figure}

Some remarks are in order, to better clarify the scope of this scenario
(the precise statements, which are the main results of the paper,
  are given in Theorems~\ref{teobaffo}, \ref{teoH}, \ref{teounasemi}
and \ref{quartocaso}).

In the first two cases ground states, as a rule, do not exist.
In case \casoI,
the presence of a tip --hence of a terminal edge--
allows the construction of ``monotone'' functions of mass $\muRp$,
that decrease away from the tip and mimic a
half-soliton on $\R^+$,
with an energy level arbitrarily close to zero (albeit strictly positive, unless
$\G$ is exactly $\R^+$): thus, in a sense,
graphs with a tip behave much like a half-line.
In case \casoII, by contrast,
the covering assumption is not compatible with ``monotone'' functions and, due to
a  rearrangement argument from $\G$ to $\R$,
no function  of mass
$\muR$  can have a negative energy on $\G$. This rigidity rules out ground states,
unless $\G$ supports a soliton, and this, in turn, occurs only when $\G=\R$
(possibly with the identification of some pairs of points,
in a way compatible with the even symmetry of a soliton).  Thus, dually,
a graph as in \casoII{} behaves much like $\R$.

The last two cases are, on the contrary, extremely nontrivial.
In \casoIII, $\G$ consists of a compact core $\K$ (with no
terminal edge) attached to a half-line, the simplest example being
the ``tadpole'' graph in Fig.~\ref{figtadpole}.

\begin{figure}
\begin{center}
\begin{tikzpicture}[scale= 1.7]  %%%%%%   tadpole
\node at (-2.5,0) [infinito]  (1) {$\scriptstyle\infty$};
\node at (1.8,0) [nodo] (2) {};
\draw [-] (1) -- (2) ;
\draw(2.2,0) circle (0.4);
\end{tikzpicture}
\caption{\footnotesize{a tadpole graph. }}
\label{figtadpole}
\end{center}
\end{figure}
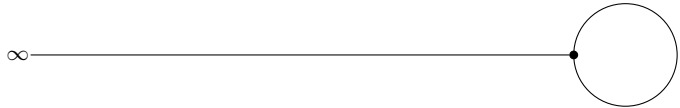

If $\mu\in (\muRp,\muR]$,
a ground state of mass $\mu$ always exists --with a \emph{strictly negative}
energy-- and this  is a completely new phenomenon.
Over $\R^+$, due to \eqref{scalE},   one
has $\elevel_{\R^+}(\mu)=-\infty$ (and no ground state)
as soon as $\mu>\muRp$: here, on the contrary,  the compact core $\K$
attached to
$\R^+$
has the effect of a \emph{stabilizer} as regards ground states:
due to $\K$, $\G$ loses dilation invariance, and high concentration
is no longer energetically convenient, which accounts for
\emph{strictly negative} (yet finite!) ground-state energy levels.

Finally, in \casoIV,   $\G$ has
no tip, no cycle covering and (being noncompact)
at least two half-lines. This case
becomes very interesting if one  \emph{further} assumes that
$\muG\!<\!\muR$, which guarantees the existence of ground states for every mass
$\mu\in [\muG,\muR]$: here, contrary to \casoIII,
a ground state exists also when  $\mu=\muG$, with
a zero energy level. A particularly interesting, and specific, feature of the case $\mu=\muG$ is the
coexistence of compact and noncompact minimizing sequences. Thus an additional
difficulty in this case is the {\em choice} of a proper sequence to work with.

Explicit examples of graphs can be
constructed
(e.g. the ``signpost" graph in Fig. \ref{figD}, as explained in Sec. 3),
where $\muG\!<\!\muR$: this extra assumption, however,
is crucial to prove the existence of ground states, and we believe that
it cannot be dropped in general.

More precisely we believe that, within case \casoIV, the \emph{sole} topology
of $\G$ is not enough, in general, to establish whether
$\muG<\muR$ or $\muG=\muR$ and, in the latter case,
whether a ground states exists, of mass $\muG$.
It is an open problem, at present, to fully understand problem \eqref{minprob}
in case \casoIV, when $\muG=\muR$.

Thus, summing up, the four cases \casoI--\casoIV{} cover all the
possible topologies of a (noncompact) metric graph $\G$. The first
two cases are extremely rigid, with
ground states being the exception rather
than the rule. Case \casoIII{} is, on the contrary, very interesting,
with  ground states in a universal range of prescribed masses,  and
one may consider such graphs as ``intermediate'' between $\R^+$ and $\R$.
Finally, case \casoIV{} is also nontrivial
(with  ground states in a whole, closed interval of prescribed masses),
but this is subordinated to the condition that $\muG<\muR$: in this
case $\G$ shows, again, an intermediate behavior between $\R^+$ and $\R$.
To conclude this discussion we wish to emphasize that all the
ground states of cases \casoIII{} and \casoIV{} (except those of mass $\muG$) have {\em negative energy}.
This is in sharp contrast with the behavior of the NLS equation on $\R$ (or $\R^n$),
where every solution with negative energy blows up in finite time (\cite{cazenave}),
and will be the object of a forthcoming paper.

\medskip

The problem of the minimization of \eqref{NLSe} under the constraint \eqref{mass} can be interpreted
according to the Gross-Pitaevskii theory for the ground state of the Bose-Einstein condensates.
Indeed, following
a series of results obtained in the last decade by several authors (see e.g. \cite{ls,lsy, ly,agt,esy1,esy2,esy3,pickl,bos,holmer}),
under some physical conditions the dynamics
of a gas of interacting identical bosons can be described through a one-body nonlinear equation, called {\em Gross-Pitaevskii equation}. More
precisely, if the particles in the gas interact in pairs, then
the resulting equation displays a {\em cubic} nonlinearity.
On the other hand, in \cite{chen}
a system of identical bosons interacting through a {\em three}-body potential was considered, and the resulting
equation was shown to be a {\em quintic} NLS.  

In actual dilute Bose-Einstein condensates, two- and three-body interactions may coexist, so that in the resulting one-body equation
both cubic and quintic terms would arise, even though normally the effect of the cubic term overwhelms the effect of the quintic, that is therefore neglected.
Concerning the sign of the nonlinear terms, it turns out to depend on the character  attractive or repulsive of the interaction among the particles,
so that it is possible to realize experimentally condensates that display either a focusing or a defocusing behaviour (\cite{wieman}).

For these reasons, and since both the functional \eqref{NLSe} and the $L^2$-norm are conserved by the evolution
driven by the quintic NLS equation
\begin{equation} \label{quintic}
i \partial_t  u (t,x) \ = \ - u'' (t,x) - |u(t,x)|^4 u (t,x),
\end{equation}
the issue of the existence of a minimizer of the constrained energy
can be interpreted as the search for the
ground state of a particular Bose-Einstein condensate
where two-body interactions are absent.

With respect to the problems currently studied as regards ground states of a Bose-Einstein condensate
and to the rigorous derivation of equation \eqref{quintic} obtained in \cite{chen},
let us stress
two remarkable differences: first, we consider a {\em focusing} nonlinearity, second,
we set the problem
on a graph. Both features have nowadays
an established experimental counterpart: on the one hand, self-concentrating condensates are currently realised \cite{wieman},
on the other hand, condensation on graph-like structures has been recently observed in \cite{lorenzo}.

While linear dynamics on quantum graphs is nowadays a well-known branch of mathematical physics (see e.g. \cite{berkolaiko, exner,post}),
its nonlinear counterpart has gained an increasing interest only quite recently.
A seminal study of nonlinear evolution equation on ramified structures appeared in \cite{ali} and then was extended to
several physical domains, involving various mathematical issues \cite{bona,vonbelow}. The study of
the evolution of solitary waves on star graphs was performed in \cite{acfn1}, while the search for ground states states was carried out in
\cite{matrasulov,sobirov,acfn2,acfn3,acfn4} for the case of star graphs, and then extended to more general
graphs (dealing with the same class considered in this paper) in \cite{ast,ast2,sven}. Stationary states and related bifurcation
were recently investigated in \cite{noja,marzuola, pelinovsky}. In particular, \cite{schneider} treats the case of periodic graphs,
which is not covered in this paper, and shows the occurrence of a bifurcation phenomenon. The integrability of
the cubic NLS on star graphs was proved in \cite{caudrelier}.
All cited papers deal with the subcritical case or even specialize to the cubic case. To our knowledge,
the present work is the first contribution to the understanding of the role of the $L^2$-criticality in the framework of graphs. The
fact that this role appears to be different from what happens on standard domains like $\R^n$ is in our opinion worth
being stressed.

The paper is organised as follows: in Section 2 we give some preliminary results and introduce the notion of critical
mass; in Section 3 we state the four main theorems. The core of the proofs  is  a result stated and proved in Section 4. Finally, in Section 5 we conclude the proofs of the
main existence results stated in Section 3.

\section{Notation and preliminary results}
\label{sec2}

It is well known (\cite{cazenave}) that when $\G=\R$, the ground-state energy level
function, as defined in \eqref{defelevel}, has a sharp transition from $0$ to $-\infty$,
corresponding to a special value $\muR$ of the mass, known as the \emph{critical mass}:
\begin{equation}
\label{Rlevel}
\elevel_\R(\mu) = \begin{cases}
 0 & \text{ if } \mu \le \mu_\R \\
  -\infty & \text{ if } \mu > \mu_\R
  \end{cases}\quad\qquad \left(\muR = \pi\sqrt 3 /2\right).
\end{equation}
Furthermore, the infimum $\elevel_\R(\mu)$ is attained
(i.e. a ground state exists in \eqref{minprob})
{\em if and only if} $\mu = \mu_\R$. Thus every ground state
$u$ (necessarily of mass $\muR$) satisfies $E(u,\R) =0$.
The ground states, called \emph{solitons}, form a quite large family:
up to phase multiplication and translations, they can be written as
\begin{equation} \label{solitoni}
\phi_\lambda(x) =  \sqrt\lambda \phi(\lambda x), \qquad \lambda>0,
\end{equation}
where
\begin{equation}
\label{solit}
\phi (x) =  \hbox{\rm{sech}}^{1/2}\left(2x/{\sqrt 3}\right).
\end{equation}

When $\G=\R^+$ (the positive half-line), the situation is similar, but with the proper
critical mass $\muRp=\muR/2$:
\[
\elevel_{\R^+}(\mu) =
\begin{cases} 0 & \text{ if } \mu \le \muRp\\
-\infty & \text{ if } \mu > \muRp
 \end{cases}
 \quad\qquad \left(\muRp = \pi\sqrt 3 /4\right).
\]
Again,
ground states of mass $\mu$ exist if and only if $\mu=\muRp$.
They are called ``half-solitons'', as
they are the restrictions to $\R^+$ of the family $\phi_\lambda$.

Thus, when $\G$ is $\R$ or $\R^+$, problem \eqref{minprob} is trivialized by these sharp transitions.

For a general noncompact graph $\G$,
the behavior of $\elevel_\G(\mu)$
is strictly related to the
 Gagliardo--Nirenberg inequality
\begin{equation}
\label{GN}
\| u\|_{L^6(\G)}^6 \le K_\G \,\| u\|_{L^2(\G)}^4\,  \| u'\|_{L^2(\G)}^2,\quad
\forall u\in H^1(\G)
\end{equation}
valid for every  noncompact graph (see \cite{ast2,tentarelli}).
The number $K_\G$ is the \emph{best constant} that one can put in \eqref{GN}, namely,
\begin{equation}
\label{defKG}
K_\G =\sup_{u\in H^1(\G) \atop u\not\equiv 0}
\frac{\| u\|_{L^6(\G)}^6}{\|u\|_{L^2(\G)}^4\cdot\|u'\|_{L^2(\G)}^2}
=\sup_{u\in H_\mu^1(\G)}\; \frac {\Vert u\Vert_{L^6(\G)}^6}
{\mu^2\cdot\Vert u'\Vert_{L^2(\G)}^2},
\end{equation}
where the last equality follows from homogeneity.

The role of this constant, in connection with the behavior of $\elevel_\G(\mu)$, is clear:
recalling \eqref{NLSe} and the definition of $H^1_\mu(\G)$ in \eqref{minprob},
using \eqref{GN} we have, for every $u\in H_\mu^1(\G)$,
\[
E(u,\G)\geq  \frac12 \|u'\|_{L^2(\G)}^2 -\frac16 K_\G \mu^2 \|u'\|_{L^2(\G)}^2 =
\frac16 \Vert u'\Vert_{L^2(\G)}^2  \bigr(3-K_\G \,\mu^2\bigr).
\]
We then see that
\begin{equation}
\label{elevelpos}
\mu^2 \le 3/K_\G  \implies   E(u,\G) \ge 0 \qquad\hbox{for all}\quad u\in H_\mu^1(\G).
\end{equation}
Note also that
\begin{equation}
\label{Epos}
\mu^2 < 3/K_\G  \implies   E(u,\G) > 0\qquad\hbox{for all}\quad u\in H_\mu^1(\G).
\end{equation}
On the other hand, if $\mu^2 >3/K_\G$, say $\mu^2 =3(1+\delta)/K_\G$ for some $\delta>0$,
we can take $u \in H_\mu^1(\G)$  close to optimality in \eqref{defKG}, say
\[
\| u\|_{L^6(\G)}^6 > \frac{K_\G}{1+\delta} \,\mu^2\,  \| u'\|_{L^2(\G)}^2,
\]
to obtain
\begin{equation}
\label{dopo}
E(u,\G)
< \f 1 6 \Vert u'\Vert_{L^2(\G)}^2  \bigr(3-\frac{K_\G}{1+\delta} \,\mu^2\bigr)  = 0.
\end{equation}
This shows that
\begin{equation}
\label{elevelneg}
\mu^2 > 3/K_\G  \implies   E(u,\G) < 0 \qquad\hbox{for some}\quad u\in H_\mu^1(\G).
\end{equation}
Now \eqref{elevelneg} and \eqref{elevelpos} justify the following definition.

\begin{definition}
\label{critmass}
The {\em critical mass} for a noncompact metric graph  $\G$ is the number
\[
\muG = \sqrt{3/ K_\G}.
\]
\end{definition}
This definition, of course, gives the correct critical mass when $\G$ is $\R$ or $\R^+$.
In general, from \eqref{Epos} and \eqref{elevelneg}, we see that $\muG$ is the precise
mass threshold, after which the ground-state energy level $\elevel_\G(\mu)$ becomes negative
(possibly $-\infty$).

\begin{remark}\label{remSC}
Any noncompact metric graph $\G$ has (at least) one unbounded
edge which, in turn, contains arbitrarily large intervals.
Therefore, any function $v\in H^1(\R)$ having \emph{compact support}  can
be regarded as an element of $H^1(\G)$, by placing the support of $v$ inside
a half-line of $\G$, and setting $v\equiv 0$ outside. Thus, in a sense, $H^1(\G)$ ``contains''
a dense subset of $H^1(\R)$.
\end{remark}

Next we notice that any noncompact $\G$ is, in a way, intermediate between $\R^+$ and $\R$  in the
sense of the following statement.

\begin{proposition}
\label{intermediate}
Let $\G$ be a noncompact graph, and let $\muG$ be the critical mass for $\G$. Then
\begin{equation}
\label{intermu}
\muRp\, \le\, \muG\, \le \,\muR
\end{equation}
or, equivalently,
\begin{equation}
\label{interK}
K_\R \,\le \, K_\G\,\le \,K_{\R^+}.
\end{equation}
Moreover, we also have
\begin{equation}
\label{interE}
\elevel_{\R^+}(\mu) \,\le \, \elevel_\G(\mu)\,\le \,\elevel_{\R}(\mu),\quad
\forall \mu>0.
\end{equation}
\end{proposition}

\begin{proof} Given  $u \in H^1(\G)$
(assume $u\geq 0$ and $u\not\equiv 0$), let $u^* \in H^1(\R^+)$ be
its decreasing rearrangement on $\R^+$ (for properties of rearrangements on graphs see \cite{ast, friedlander}).
Since
\[
\| (u^*)'\|_{L^2(\R^+)}^2\leq\| u'\|_{L^2(\G)}^2,
\quad
\| u^*\|_{L^p(\R^+)}^p=\| u\|_{L^p(\G)}^p
\quad\forall p,
\]
the quotient for $u$ in \eqref{defKG} does not exceed the same quotient
for $u^*$ (over $\R^+$): since the latter is bounded by $K_{\R^+}$, we obtain
that
$K_\G\le K_{\R^+}$. In the same way, we have $E(u^*,\R^+)\leq E(u,\G)$, and hence
$\elevel_{\R^+}(\mu)\leq E(u,\G)$, where $\mu:=\| u\|_{L^2(\G)}^2=\| u^*\|_{L^2(\R^+)}^2$.
By the arbitrariness of $u$, we obtain the first inequality in \eqref{interE}.

Finally, let $H^1_{\mu,c}$ denote the set of all $u\in H^1_\mu(\R)$ with compact support.
By a density argument, we have
\[
\elevel_\R(\mu)
=\inf_{u\in H^1_{\mu,c}(\R)} E(u,\R)
\geq \inf_{u\in H^1_\mu(\G)} E(u,\G)
=\elevel_\G(\mu)
\]
(the inequality follows from Remark~\ref{remSC}). This proves the second inequality
in \eqref{interE}; the first inequality
in \eqref{interK} is proved in the same way, working with the supremum in \eqref{defKG}.
\end{proof}

After this discussion, we summarize in the next proposition the properties that hold for a generic noncompact graph,
without any additional assumption.

\begin{proposition}
\label{banali}
Let $\G$ be a noncompact metric graph.
\begin{enumerate}
\item[(i)] If $\mu\le \muG$, then $\elevel_\G(\mu) = 0$, and is never attained when $\mu<\muG$.
\item[(ii)] If $\mu >\muG$, then $\elevel_\G(\mu) < 0$ (possibly $-\infty$).
\item[(iii)] If $\mu > \muR$, then $\elevel_\G(\mu) = -\infty$.
\end{enumerate}
\end{proposition}

\begin{proof} When $\mu \le \muG$, \eqref{elevelpos} shows that $\elevel_\G(\mu) \ge 0$.
On the other hand, we infer from \eqref{interE} that
$\elevel_\G(\mu)\leq \elevel_\R(\mu)$,
and the latter is zero due to \eqref{Rlevel}, since $\mu\leq \muG\leq \muR$ by \eqref{intermu}.
Moreover,
by \eqref{Epos} we
see that $\elevel_\G(\mu)$ is not attained when $\mu < \muG$.
This proves  (i).

By \eqref{elevelneg}, one immediately obtains (ii).

Finally, to prove (iii), observe that
$\elevel_\G(\mu)\leq \elevel_\R(\mu)=-\infty$, due
to \eqref{intermu} and \eqref{Rlevel}.

\end{proof}

\begin{corollary}\label{cor1} A necessary condition for the existence
of a ground state of mass $\mu$ in \eqref{minprob} is that
$\mu\in[\muG,\muR]$.
\end{corollary}

\section{Statement of the main results}

In this section we state the main results of the paper
(Theorems \ref{teobaffo}--\ref{quartocaso}), thus
providing a precise and formal setting
for the four possible cases \casoI--\casoIV{}, that were informally described in the Introduction.

The following theorem covers case \casoI, represented in Fig. \ref{figbaffo}.

By a ``terminal point'' (or ``tip'')
we mean a point $x$, in the metric graph $\G$, that corresponds to a
\emph{vertex of degree one} in the underlying (combinatorial) graph.
Usually, $x$ is one of the two endpoints of a \emph{bounded} edge
attached to the rest of $\G$ only at the other endpoint, as a pendant: the only
exception is when $\G$ consists of exactly one unbounded edge (i.e. when
$\G=\R^+$), in which case the tip $x$ is the origin of the half-line.
We point out that the $\infty$-point of any half-line of $\G$ (though being
a vertex of degree one in the underlying combinatorial graph) is \emph{not}
a terminal point, since it is not a point of $\G$ (as a metric graph).

\begin{theorem}[graphs with a tip]
\label{teobaffo} Let $\G$ be a noncompact metric graph having at least one
terminal point (a tip). Then $\muG= \muRp$.
When $\mu \in (\muRp, \muR]$, $\elevel_\G(\mu) = -\infty$. When $\mu = \muRp$,
$\elevel_\G(\mu)=0$ and it is attained if and only if $\G$ is isometric to a half-line.
\end{theorem}

The next theorem covers case \casoII, when $\G$ admits a cycle covering (see Fig. \ref{figH}).

Here and throughout, by ``cycle'' we mean either a \emph{loop} (a homeomorphic
image of $\mathcal S^1$) or, by extension, an unbounded path that joins two
(necessarily distinct) $\infty$-points of
$\G$. In the former case the cycle corresponds to a closed path
in the underlying combinatorial graph (i.e. it is a ``cycle'' in the usual sense of
graph theory) whereas, in the latter case, it does not (the two notions would essentially coincide,
however, if properly reformulated for the one-point compactification of $\G$).
Alternatively, in the underlying combinatorial graph,
one might identify all the $\infty$-points of $\G$ into a unique, special
vertex (of degree equal to the number of half-lines of $\G$): in this way, also cycles
of the second type would be usual cycles in the graph-theoretic sense.

The existence of a cycle covering is equivalent (see \cite{ast, ast2}) to a
property of $\G$, called ``assumption (H)'', first identified in \cite{ast} as a topological
obstruction to the existence of ground states in the subcritical cases. In particular, this assumption
is incompatible with the presence of tips and forces the graph to have at least two half-lines.

\begin{theorem}[graphs with a cycle covering]
\label{teoH} Let $\G$ be a noncompact metric graph that admits a cycle covering.
Then  $\muG= \muR$. The infimum $\elevel_\G(\mu)$ is attained if and only if
$\mu=\muR$ and $\G$ is $\R$ or a ``tower of bubbles'' (one of the special graphs described in Example 2.4 of \cite{ast}).
\end{theorem}

A different behaviour occurs if one considers graphs without terminal points and with one half-line only: the critical mass turns out to coincide with $\mu_{\R^+}$, but
the ground state energy level remains finite if the mass does not exceed $\mu_\R$. Furthermore, in the interval $(\mu_{\R^+}, \mu_\R]$ a ground state exists and it has negative
energy.

\begin{theorem}[graphs with one half-line]
\label{teounasemi}
Let $\G$ be a noncompact metric graph having exactly one half-line
and no terminal point. Then $\muG=\muRp$. The infimum $\elevel_\G(\mu)$ is attained if and only if $\mu \in(\muRp,\muR]$.
\end{theorem}

The last theorem deals with the remaining cases, under the additional hypothesis $\mu_\G < \mu_\R$. Notice that, for such graphs,
the interval of masses where a ground state exists, is closed.

\begin{theorem}
\label{quartocaso}
Let $\G$ be a noncompact metric graph with no terminal point, having at least two half-lines and admitting no cycle covering.
If $\mu_\G <\mu_\R$, then for every $\mu\in [\muG,\muR]$, the infimum $\elevel_\G(\mu)$ is attained.
\end{theorem}

The class of graphs satisfying the assumptions of Theorem \ref{quartocaso} is not empty, since it contains, for instance, the ``signpost" graph $\G$ of Fig. \ref{figD}. To see this we take a 
soliton $\phi$ (necessarily of mass $\muR$) as defined in \eqref{solitoni} and we apply the procedure detailed in Sec. 3 of \cite{ast2}. This produces
a function $u \in H^1_{\mu_\R} (\G)$ such that
$E(u, \G) < E (\phi, \R) =0$. Therefore, $\elevel_\G(\muR) < 0$ and,  by Proposition \ref{banali},
we conclude that $\muG <\muR$.
\medskip

\begin{remark}\label{remGN}
The last four theorems also provide an answer to the question of
existence of extremal functions for the Gagliardo-Nirenberg inequality
\eqref{GN}. The key observation is that,
for any noncompact graph $\G$, the following two conditions are in fact equivalent:
\begin{itemize}
\item[(i)] there exists $u\in H^1(\G)$, $u\not\equiv 0$, achieving equality in \eqref{GN};
\item[(ii)] the infimum $\elevel_\G(\muG)$ is attained by a ground state of mass $\muG$.
\end{itemize}
Indeed, by homogeneity, a function extremal for \eqref{GN} can be supposed to have
mass $\muG$: on the other hand, since $K_\G\,\muG^2=3$,
optimality in \eqref{GN} combined with $\Vert u\Vert_{L^2}^2=\muG$
is equivalent to $E(u,\G)=0$, i.e. to $u$ being a ground state, by Proposition~\ref{banali}.
\end{remark}

We end this section with the short proofs of the first two theorems. The last
two theorems are, on the contrary, much more involved, and to their
proofs are devoted the last sections of the paper.

\begin{proof}[Proof of Theorem \ref{teobaffo}] Let $\mu >\muRp$. For every $\eps >0$, there exists $u \in H_\mu^1(\R^+)$ with compact support such that
\begin{equation}
\label{nearly}
\frac {\| u\|_{L^6(\R^+)}^6} {\mu^2\cdot\| u'\|_{L^2(\R^+)}^2} \ge K_{\R^+} -\eps.
\end{equation}
Replacing $u$ by $u_\lambda(x) = \sqrt \lambda u(\lambda x)$ with $\lambda$ large,
we can assume that the support of $u_\lambda$ is contained in an interval shorter than a terminal edge of $\G$.
Then the function $u_\lambda$  can be seen as an element of $H_\mu^1(\G)$ (place the support of $u$ on a terminal edge
of $\G$, and set $u_\lambda\equiv 0$ elsewhere on $\G$). Consequently, the quotient in the
preceding inequality does not exceed $K_\G$. Thus,
for every $\eps >0$, $K_\G \ge K_{\R^+} - \eps$ and, recalling \eqref{interK},
we obtain $K_\G = K_{\R^+}$, that is, $\muG = \muRp$.

To prove that
$\elevel_\G(\mu)= -\infty$, notice that  since $K_\G = K_{\R^+}$, \eqref{nearly}
readily implies, for $\eps$ small, that  $E(u_\lambda,\G) <0$ (as in \eqref{dopo}).
Therefore
\[
E(u_\lambda,\G) = E(u_\lambda,\R^+) = \lambda^2 E(u,\R^+) \to -\infty
\]
as $\lambda \to +\infty$. Thus, $\elevel_\G(\mu) = -\infty$.

Finally,  let  $\mu = \muRp$.
If $\G=\R^+$, plainly $\elevel_{\R^+}(\mu_{\R^+})$ is attained by the half-solitons. Conversely,
assume that there exists $u \in H^1_{\mu_{\R^+}}({\G})$ such that $E(u,\G) = \elevel_\G(\muRp) = 0$. Its decreasing rearrangement $u^*$ is in $H^1_ {\mu_{\R^+}}({\R^+})$ and satisfies $E(u^*,\R^+) \le E(u,\G) = 0$. Since
$\elevel_{\R^+}(\muRp) = 0$, the function $u^*$ must be a half-soliton. From $E(u,\G) = E(u^*,\R^+)$ we deduce (via Proposition 3.1 of \cite{ast}) that
almost every point in the range of $u$ has exactly one preimage.
As $u$ and $u^*$ are equimeasurable, it follows that $u$ is injective, which means that $u$ is supported on a subset of a half-line.
But as  $E(u,\G) = 0$, also $u$ must be a half-soliton and this is possible only if there are no vertices other than the ends of the half-line, i.e., $\G$ is (isometric to) $\R^+$.
\end{proof}

\begin{proof}[Proof of Theorem \ref{teoH}] Take any nonnegative $v\in H^1(\G)\setminus \{0\}$,
and let $\widehat v\in H^1(\R)$ denote its symmetric rearrangement on $\R$.
Since $\G$ admits a cycle covering (namely, it satisfies assumption (H)), we have (see Proposition 3.1 in \cite{ast})
\[
\| \widehat{v}'\|_{L^2(\R)} \le  \| v'\|_{L^2(\G)}
\]
and therefore
\[
\frac {\| v\|_{L^6(\G)}^6} {\| v\|_{L^2(\G)}^4 \cdot \| v'\|_{L^2(\G)}^2} \le
\frac {\|\widehat{v}\|_{L^6(\R)}^6} {\|\widehat{v}\|_{L^2(\R)}^4 \cdot \| \widehat{v}'\|_{L^2(\R)}^2}
\le K_\R,
\]
showing that $K_\G \le K_\R$. By \eqref{interK} we obtain $K_\G = K_\R$, namely $\mu_\G = \muR$.

If $\mu\ne \muR =\muG$, Proposition \ref{banali} shows that $\elevel_\G(\mu)$ is not attained.

If  $\mu = \muR$, a soliton can be placed on $\R$ or on any tower of bubbles (see \cite{ast}), showing that those graphs
carry a ground state. Conversely, assume that there exists $u \in H^1_{\muR}({\G})$ such that $E(u,\G) = \elevel_\G(\muR) = 0$. Since $\G$ satisfies assumption (H), the symmetric rearrangement $\widehat u$ is in $H^1_ {\muR}(\R)$ and satisfies $E(\widehat u,\R) \le E(u,\G) = 0$. Since
$\elevel_{\R}(\muR) = 0$, the function $\widehat u$ is a soliton.
Furthermore, since $E(u,\G) = E(\widehat u, \R)$, almost every point in the range of $u$ has exactly two preimages.
These features show, as in Theorem 2.5 of \cite{ast}, that $\G$ must be one of the special graphs of \cite{ast}.
\end{proof}

\section{A general existence argument}

We now present a general existence result for ground states, which is the
common core of the proofs of Theorems~\ref{teounasemi} and \ref{quartocaso},
as long as $\muG<\mu\leq\muR$.

\begin{proposition}
\label{mainprop}
Let $\G$ be a noncompact metric graph with no terminal point,
such that $\mu_\G <\mu_\R$. For every $\mu \in (\muG,\muR]$ the infimum $\elevel_\G(\mu)$  is attained.
\end{proposition}

The proof is quite involved, and will rely on the following three lemmas.

\begin{lemma}
\label{weakzero}
Assume $\G$ is noncompact and
let $v_n\in H^1(\G)$ be a sequence of  functions such that
$v_n \rightharpoonup 0$ in $H^1(\G)$.  
Then, as $n \to \infty$,
\begin{equation}
\label{weaklevel}
E(v_n,\G) \ge  \frac12\left( 1 - \frac{\|v_n\|_{L^2(\G)}^4}{\mu_\R^2}\right)\|v_n'\|_{L^2(\G)}^2 +o(1).
\end{equation}
\end{lemma}

\begin{proof}
Let $\K$ be the compact core of $\G$ (i.e., the compact metric graph obtained from
$\G$ by removing the interior of every halfline)  and set
\[
\eps_n := \max_{x\in \K} |v_n(x)|,
\qquad
w_n(x) = \begin{cases}
v_n(x)+\eps_n &\text{if $v_n(x)\leq-\eps_n$}\\
0&\text{if $|v_n(x)|<\eps_n$}\\
v_n(x)-\eps_n &\text{if $v_n(x)\geq\eps_n$}
\end{cases}
\]
Since $v_n \rightharpoonup 0$ in $H^1(\G)$, we have $v_n \to 0$ in $L^{\infty}_{\rm loc}(\G)$
and hence $\eps_n\to 0$. Moreover, as  $\Vert v_n-w_n\Vert_{L^\infty}\leq\eps_n$, from
equiboundedness in $L^2(\G)$ we have
\begin{equation}
\label{cof}
\|v_n -w_n\|_{L^6(\G)} \to 0.
\end{equation}
Now let $\HH$ be an arbitrary halfline of $\G$. By choosing a coordinate $t\geq 0$ on it,
we can identify $\HH$ with $[0,+\infty)$ and,
since $w_n\equiv 0$ on $\K$ and $\HH$ is attached to $\K$ at $t=0$,
we have $w_n(0)=0$. Setting $w_n(t)\equiv 0$ for $t<0$, the restriction of $w_n$ to
$\HH$ can be seen as a function in $H^1(\R)$ and, as such, it must obey the following
Gagliardo-Nirenberg iequality:
\begin{equation*}
\|w_n\|_{L^6(\HH)}^6
\le 3 \frac{\| w_n\|_{L^2(\HH)}^4}{\mu_\R^2}
\| w_n'\|_{L^2(\HH)}^2
\leq
3 \frac{\| w_n\|_{L^2(\G)}^4}{\mu_\R^2}
\| w_n'\|_{L^2(\HH)}^2.
\end{equation*}
Summing over all the halflines of $\G$, since $w_n\equiv 0$ on $\K$ we obtain
 \begin{equation*}
\|w_n\|_{L^6(\G)}^6
\le 3 \frac{\| w_n\|_{L^2(\G)}^4}{\mu_\R^2}
\| w_n'\|_{L^2(\G)}^2
\end{equation*}
(just as it would be if we had $\muG=\muR$).
Since $\| w_n'\|_{L^2(\G)}^2 \le \| v_n'\|_{L^2(\G)}^2$
and $\| w_n\|_{L^2(\G)}^2 \le \| v_n\|_{L^2(\G)}^2$, from \eqref{cof} and the
previous inequality we obtain,
\[
\|v_n\|_{L^6(\G)}^6 +o(1) = \|w_n\|_{L^6(\G)}^6 \le  3 \frac{\|v_n\|_{L^2(\G)}^4 }{\mu_\R^2}
\| v_n'\|_{L^2(\G)}^2,
\]
and \eqref{weaklevel} follows immediately from the definition of $E(v_n,\G)$.
\end{proof}

\begin{lemma}
\label{limit}
Let $\G$ be a noncompact graph,
and take $\mu \in [\mu_\G, \mu_\R]$.
be a nonnegative minimizing sequence for $E(\, \cdot\,,\G)$ such that
$
u_n \rightharpoonup u\quad\text {in } H^1(\G)
$, for some $u\in  H^1(\G)$.

If $u\not\equiv 0$, then $u\in H^1_\mu(\G)$ and $u$ is a minimizer.
\end{lemma}

\begin{proof}
Passing to a subsequence, we may assume that
$u_n(x) \to u(x)$ a.e. in $\G$.
By a standard use of the Brezis--Lieb Lemma (\cite{BL}), as in \cite{ast2},
\[
E(u_n,\G) = E(u_n-u,\G) + E(u,\G) + o(1)
\]
as $n\to \infty$. Since $u_n-u \rightharpoonup 0$ in $H^1(\G)$,
from Lemma \ref{weakzero} applied
with $v_n=u_n-u$ we obtain
\[
E(u_n-u,\G) \ge \frac12\left(1 -  \frac{\| u_n - u\|_{L^2(\G)}^4}{\mu_\R^2}\right)  \| u_n' - u'\|_{L^2(\G)}^2 + o(1)
\]
as $n\to \infty$.
Therefore
\[
E(u_n,\G) \ge E(u, \G) + o(1),
\]
that is,
\[
E(u,\G) \le \elevel_\G(\mu).
\]
Now, by semicontinuity, we have $m:=\Vert u\Vert_{L^2(\G)}^2\leq\mu$, and $m>0$ by assumption.
 If $m<\mu$, then
\[
E\left(\sqrt{\frac{\mu}{m}}u,\G\right) =\frac{\mu}{m}\frac12\int_\G |u'|^2\dx - \left(\frac{\mu}{m}\right)^3\frac16 \int_\G |u|^6\dx <
\frac{\mu}{m}E(u,\G)\le \elevel_\G(\mu),
\]
since $\mu/m>1$, $\Vert u\Vert_{L^6(\G)}>0$ and $\elevel_\G(\mu) \le 0$. This contradicts the definition of $\elevel_\G(\mu)$. Then it must be $m=\mu$,
namely $u$ is the required minimizer. It is also easy to see that the convergence of $u_n$ to $u$ is strong in $H^1(\G)$.
\end{proof}

The next lemma establishes a crucial modification of the Gagliardo--Nirenberg inequality. Its proof is quite long, and is therefore
split in a series of steps.

\newcommand\exM{\theta}

\begin{lemma}[modified G-N inequality]\label{modGN} Assume $\G$ is noncompact and
has no terminal point, and let $u\in H^1_\mu(\G)$ for some $\mu\in (0,\muR]$. Then
there exists a number $\exM\in [0,\mu]$ such that
\begin{equation}
\label{sharpineq}
\|u\|_{L^6(\G)}^6   \le  3\left(\frac{\mu-\exM}{\muR}\right)^2  \|u'\|_{L^2(\G)}^2  +C\exM^{1/2},
\end{equation}
where $C>0$ is a constant that depends only on $\G$.
\end{lemma}
\begin{proof} Replacing $u$ with $|u|$, we may assume that $u\geq 0$ and
$u\not\equiv 0$.
\medskip

\noindent\emph{Step 1. }
There exist $\ell>0$ (depending only on $\G$)
 and a function $\psi\in H^1((-\infty,\ell])$, such that
\begin{enumerate}
\item $\int_{-\infty}^\ell |\psi|^2\dx = \int_\G |u|^2\dx = \mu$ 
\item $\int_{-\infty}^\ell |\psi|^6\dx = \int_\G |u|^6\dx$,
\quad\text{while}
\quad
$\int_{-\infty}^\ell |\psi'|^2 \dx \le  \int_\G |u'|^2 \dx$;
\item $\psi$ is nonnegative on $(-\infty,\ell]$, and nonincreasing on $[0,\ell]$;
\end{enumerate}

\proof[Proof of step 1] If $\G$ has at least a loop (a homeomorphic image of $\mathcal S^1$),
let $2\ell$ denote the length of the shortest loop; otherwise, let $\ell:=1$.
Let us denote by $M$ the maximum of $u$ on $\G$, and by $x_0$ a point of $\G$ such that
$u(x_0) = M$: since $\G$ is noncompact and
connected, there is a path $\Gamma$ in $\G$ that joins an $\infty$-point
of $\G$ to $x_0$.

Since $x_0$ is not a terminal point of $\G$, the path $\Gamma$ can be prolonged
 beyond $x_0$, to a longer path that crosses $x_0$. Two cases are possible:
 (i) $\Gamma$ can be prolonged beyond $x_0$ for a length
  $\ell$, or (ii) a self intersection occurs (i.e. a loop is created) before
  an extra length of $\ell$ has been traveled.

\smallskip

  In case (i), let $\gamma$ denote the new path (of length $\ell$, starting at $x_0$)
  used to prolong $\Gamma$ (note: $\gamma$ shares with $\Gamma$ only the point $x_0$).
  We then define the function $\psi:[0,\ell]\to\R$ as the decreasing rearrangement of
  $u_\gamma$ (i.e. $u$ restricted to $\gamma$), so that
  \[
  \psi(0)=M,\quad \int_0^\ell |\psi'|^2\,dx\leq
  \int_\gamma |u'|^2\,dx,
  \quad
  \int_0^\ell |\psi|^p\,dx=
  \int_\gamma |u|^p\,dx\quad\forall p.
  \]
Now assume, for a while,
that the metric graph $\overline{\G\setminus\gamma}$
(the closure of $\G\setminus\gamma$) is \emph{connected}:
in this case, we can consider the function $u^*:[0,+\infty)\to\R$,
defined as the decreasing rearrangement of  the
restricted function $u_{\overline{\G\setminus\gamma}}$, and observe that,
as before,
\[
  u^*(0)=M,\quad \int_0^\infty |(u^*)'|^2\,dx\leq
  \int_{\G\setminus \gamma} |u'|^2\,dx,
  \quad
  \int_0^\infty |u^*|^p\,dx=
  \int_{\G\setminus\gamma} |u|^p\,dx\quad\forall p.
  \]
Then the construction of $\psi$ is easily completed, by letting $\psi(x)=u^*(-x)$
for every $x<0$.

In general, though, $\overline{\G\setminus\gamma}$
can  be disconnected (which would prevent the
use of the monotone rearrangement): one of its connected component $\G_0$,
however, contains the original path $\Gamma$ and,
 since $u\geq 0$ and $\Gamma$ contains a half-line along which $u$ tends to zero,
the range of $u_\Gamma$ ($u$ restricted to $\Gamma$)
is the interval $[0,M]$, that is, the full range of $u$.
Therefore,
by a graph--surgery procedure, every other connected component $\G_j$ ($j>0$)
of $\overline{\G\setminus\gamma}$
can still be ``reattached'' to $\Gamma$ (hence to $\G_0$) as follows:
if $\G_j$ was originally attached to $\gamma$ at a point $y_j$ (not necessarily unique),
 take $z_j\in\Gamma$ such that $u(z_j)=u(y_j)$, and attach $y_j$ at $z_j$
 (i.e., identify $y_j$ with $z_j$).
By this trick (which preserves the continuity of $u$, its integral norms and those of $u'$)
the restricted function
 $u_{\overline{\G\setminus\gamma}}$ can now be seen
 as if defined on a connected graph, and the theory of rearrangements becomes available:
 then, the construction of $\psi$ can be completed as before.

\smallskip

In case (ii), let $x_1\in\G$ be the point where the self-intersection occurs, and let $\gamma$
denote the new-added curve, from $x_0$ to $x_1$. The length of $\gamma$ is at most $\ell$,
otherwise we would be in case (i), and hence (since $\G$ has no loop of length smaller than
$2\ell$) we see that $x_1\in\Gamma$ (in other words, the self-intersection occurs along
$\Gamma$, not along $\gamma$, which is a simple arc). For the same reason,
the length of the complementary arc $\gamma'\subset\Gamma$, from $x_1$ to $x_0$,
 is \emph{at least} $\ell$. But then, observing that $\Gamma':=(\Gamma\setminus\gamma')\cup\gamma$
is still a path from an $\infty$-point of $\G$ to $x_0$, the proof can be completed
as in case (i),  with $\Gamma'$ and (a suitable portion of) $\gamma'$ (the latter prolonging the former,
for a length of  $\ell$) now playing the roles of the original $\Gamma$ and $\gamma$.

\medskip

\noindent\emph{Step 2. }
There exist $x_0 \in [\ell/2,\ell)$ and $v\in H^1(\R^+)$ such that, defining
\begin{equation}
\label{defM}\exM:= \frac12\int_{x_0}^\ell \psi^2\dx,
\end{equation}
there hold:
\begin{itemize}
\item[i)] $v(0) = \psi(0)$;
\item[ii)] $\displaystyle\int_0^\infty |v|^2\dx = \int_0^\ell |\psi|^2\dx -  \exM$;
\item[iii)] $\displaystyle \int_0^\infty |v'|^2\dx \le \int_0^\ell |\psi'|^2\dx  + C\exM^{1/2}$;
\item[iv)] $\displaystyle \int_0^\infty |v|^6\dx \ge \int_0^\ell |\psi|^6\dx  - C\exM$;
\end{itemize}
the constant $C>0$ depends only on  $\G$.

\medskip

\proof[Proof of step 2.] Here we shall work with $\psi$ restricted to the interval $[0,\ell]$.

We first prove the existence of $x_0\in [\ell/2,\ell)$ such that
\begin{equation}
\label{eq55}
|\psi(x_0)|^4 \leq \frac {64 m^{1/2}}{\ell^2}\left(\int_{x_0}^\ell |\psi|^2\,dx\right)^{3/2},\qquad
m:=\int_{\ell/2}^\ell |\psi|^2\,dx.
\end{equation}
To see this, let $F(x)=\int_x^\ell |\psi|^2\dx$. If the inequality in \eqref{eq55} were false
for every $x_0\in [\ell/2,\ell)$, we would have
\[
-F'(x) = \psi(x)^2  > \frac {8 m^{1/4}}{\ell} \left(\int_x^\ell |\psi|^2\,dx\right)^{3/4}=  \frac {8 m^{1/4}}{\ell} F(x)^{3/4}
\]
for every $ x\in [\ell/2,\ell)$. Therefore,
\[
-\left(F(x)^{1/4}\right)' >\frac {2m^{1/4}}{\ell} \qquad\forall x\in [\ell/2,\ell)
\]
and, since $F(\ell)=0$, integration over $(\ell/2,\ell)$ yields
\[
F(\ell/2)^{1/4}> \frac {2 m^{1/4}}{\ell} \cdot \frac \ell 2 = m^{1/4},
\]
which is clearly a contradiction due to how $m$ and $F$ were defined.

We now take a point $x_0$ satisfying \eqref{eq55}, and define $\exM$
as in \eqref{defM}. If $\exM =0$, then $\psi$ vanishes on $[x_0,\ell]$; in this case we define $v$ by extending
$\psi$ to $0$ on $[\ell,+\infty)$.
If $\theta \ne 0$, we define $v: [0,+\infty) \to \R$ as
\[
v(x)=\begin{cases}
\psi(x) & \text{if $0\leq x \leq x_0$,}\\
\psi(x_0)e^{-\lambda (x-x_0)}&\text{if $x>x_0$}
\end{cases}
\]
where
\[
\lambda:=\frac {|\psi(x_0)|^2}{2\exM}.
\]
Clearly $v\in H^1(\R^+)$ and $v(0) = \psi(0)$, so that $i)$ is satisfied. Next,
\begin{align}
\int_0^\infty |v|^2 \dx &= \int_0^{x_0} |\psi|^2 \dx + |\psi(x_0)|^2 \int_{x_0}^\infty e^{-2\lambda(x-x_0)} \dx = \int_0^{x_0} |\psi|^2 \dx + \frac {|\psi(x_0)|^2}{2\lambda}\nonumber \\
&=  \int_0^{x_0} |\psi|^2 \dx +   \exM =  \int_0^\ell |\psi|^2 \dx -  \exM,
\end{align}
and $ii)$ is proved. Similarly,
\begin{align}
\int_0^\infty |v'|^2 \dx &= \int_0^{x_0} |\psi'|^2 \dx + \lambda^2|\psi(x_0)|^2 \int_{x_0}^\infty e^{-2\lambda(x-x_0)} \dx = \int_0^{x_0} |\psi'|^2 \dx + \frac {\lambda|\psi(x_0)|^2}{2} \nonumber \\
&=  \int_0^{x_0} |\psi'|^2 \dx + \frac{|\psi(x_0)|^4}{4\exM} \le  \int_0^\ell |\psi'|^2 \dx +\frac{32 (2m)^{1/2}}{\ell^2}  \exM^{1/2} \nonumber \\
&\le  \int_0^\ell |\psi'|^2 \dx + C\exM^{1/2}
\end{align}
since $m\le \muR = \pi\sqrt 3 /2$ by \eqref{Rlevel}, while $\ell$ depends only on $\G$.
 This proves $iii)$.
Finally,
\begin{align}
\label{v6}
\int_0^\infty |v|^6 \dx &= \int_0^{x_0} |\psi|^6 \dx + |\psi(x_0)|^6 \int_{x_0}^\infty e^{-6\lambda(x-x_0)} \dx = \int_0^{x_0} |\psi|^6 \dx + \frac {|\psi(x_0)|^6}{6\lambda} \nonumber \\
& \ge   \int_0^\ell |\psi|^6 \dx - \int_{x_0}^\ell |\psi|^6\dx.
\end{align}
Since $\psi$ is decreasing (on $[0,\ell]$) and $x_0\geq \ell/2$, we have
\[
\int_{x_0}^\ell |\psi|^6\,dx \le |\psi(x_0)|^4 \int_{x_0}^\ell |\psi|^2\dx \le   2\exM|\psi(\ell/2)|^4
\]
and
\[
|\psi(\ell/2)|^2\leq \frac 2 {\ell} \int_0^{\ell/2} |\psi|^2\,dx
< \frac {2\mu}\ell \le C
\]
as above,
which plugged into the previous inequality gives
\[
\int_{x_0}^\ell |\psi|^6\,dx \le C\exM.
\]
Therefore \eqref{v6} reads
\[
\int_0^\infty |v|^6 \dx \ge  \int_0^\ell |\psi|^6 \dx -C\exM,
\]
and this concludes the proof.

\medskip

\noindent\emph{Step 3.} Combining $\psi$ and $v$, we now define
\[
w(x)=\begin{cases}
\psi(x) & \text{if $x \le 0$,}\\
v(x) &\text{if $x>0$.}
\end{cases}
\]
Clearly $w\in H^1(\R)$ and, by the properties of $\psi$ and $v$,
\[
\int_\R |w|^2\dx = \int_{-\infty}^0 |\psi|^2\dx + \int_0^\infty |v|^2\dx =  \int_{-\infty}^\ell |\psi|^2\dx-  \exM =
\int_\G|u|^2\dx -\exM =\mu - \exM.
\]

By the Gagliardo--Nirenberg inequality \eqref{GN},
\begin{align}
\label{GNw}
\|w\|_{L^6(\R)}^6 &\le K_\R \left(\mu -\exM\right)^2 \|w'\|_{L^2(\R)}^2 = K_\R\muR^2 \left(\frac{\mu-\exM}{\muR}\right)^2  \|w'\|_{L^2(\R)}^2\nonumber \\
& = 3\left(\frac{\mu-\exM}{\muR}\right)^2  \|w'\|_{L^2(\R)}.
\end{align}
Now, still from the properties of $\psi$ and $v$,
\begin{align}
\label{stima6}
\|w\|_{L^6(\R)}^6 &= \int_{-\infty}^0 |\psi|^6\dx + \int_0^\infty |v|^6\dx \nonumber \\
&\ge \int_{-\infty}^0 |\psi|^6\dx + \int_0^\ell |\psi|^6 \dx - C\exM
= \|u\|_{L^6(\G)}^6  -C\exM
\end{align}
and
\begin{align}
\label{stimader}
\|w'\|_{L^2(\R)}^2 &= \int_{-\infty}^0 |\psi'|^2\dx + \int_0^\infty |v'|^2\dx \nonumber \\
& \le \int_{-\infty}^0 |\psi'|^2\dx + \int_0^\ell |\psi'|^2 \dx + C\exM^{1/2}
\le \|u'\|_{L^2(\G)}^2  +C\exM^{1/2}.
\end{align}
Inserting \eqref{stima6} and \eqref{stimader} in \eqref{GNw} we obtain
\[
 \|u\|_{L^6(\G)}^6  -C\exM \le  3\left(\frac{\mu-\exM}{\muR}\right)^2  \left( \|u'\|_{L^2(\G)}^2  +C\exM^{1/2}\right).
\]
Rearranging terms and observing that $\exM \le \exM^{1/2}\muR^{1/2} = C\exM^{1/2}$,
one obtains \eqref{sharpineq}.

\end{proof}

We are now in a position to prove Proposition~\ref{mainprop}.

\proof[Proof of Proposition \ref{mainprop}]
Fix a mass $\mu\in(\muG,\muR]$. We know from Proposition~\ref{banali}
that $\elevel_\G(\mu)<0$ (possibly $-\infty$, until proven otherwise), thus we only consider
functions $u\in H^1_\mu(\G)$ such that $E(u,\G) \le -\alpha $ for some fixed $\alpha>0$
that depends on $\mu$.
Let $u$ be any of these functions:
since $\mu \le \muR$, Lemma \ref{modGN} applies, and \eqref{sharpineq}  yields
\[
\|u\|_{L^6(\G)}^6   \le  3\left(1-\frac{\exM_u}{\muR}\right)^2  \|u'\|_{L^2(\G)}^2  +C\exM_u^{1/2}
\]
for some  $\exM_u\in (0,\mu)$. We have denoted by $\exM_u$ the constant $\exM$ appearing in \eqref{sharpineq} to
stress its dependence on $u$.
Rearranging terms we obtain
\[
3\frac{\exM_u}{\muR}\left(2-\frac{\exM_u}\muR \right)\|u'\|_{L^2(\G)}^2-C\exM_u^{1/2} \le  6E(u,\G) \le -6\alpha,
\]
and (since $\exM_u <\mu\leq\muR$) this shows that $\exM_u$ is bounded away from zero (in terms of
$C$ and $\alpha$, uniformly with respect to $u$).
Once this is established, the same inequality also shows that $ \|u'\|_{L^2(\G)}$ is bounded from above,
in terms of
$C$ and $\alpha$.
 We have thus proved that whenever $E(u,\G) \le \alpha <0$, the $L^2$ norm of $u'$ on $\G$
is uniformly bounded (the bounding constant depends only on $\G$ and $\mu$, via $C$ and $\alpha$).
By the Gagliardo--Nirenberg inequality the same holds for $\|u\|_{L^p(\G)}$, for every $p\in [2,+\infty]$.
In particular, $\elevel_\G(\mu)$ is finite.

Finally, we show that $\elevel_\G(\mu)$ is attained.

Let $u_n \in H_{\mu}^1(\G)$ be a minimizing sequence for $E(\,\cdot\, , \G)$. By the preceding argument we can assume that $\|u_n'\|_{L^2(\G)}$,
$\|u_n\|_{L^6(\G)}$ and $\|u_n\|_{L^\infty(\G)}$ are bounded independently of $n$. Up to subsequences,
$u_n\rightharpoonup u$ in $H^1(\G)$, as well as $u_n\to u$ in $L^q_{\rm loc}(\G)$ for every $q\in [2,\infty]$.

If $u \equiv 0$, by Lemma \ref{weakzero},
\[
E(u_n, \G) \ge \frac12\left(1 -  \frac{\mu^2}{\mu_\R^2}\right)  \| u_n'\|_{L^2(\G)}^2 + o(1) \ge o(1),
\]
since $\mu \le \mu_\R$. This implies $\elevel_\G(\mu) \ge 0$, which is false. Thus, the weak limit $u$ is not identically zero and, by Lemma \ref{limit}, is the required minimizer.
\bigskip

\section{Proof of the main results}

In this section, building on Proposition~\ref{mainprop},  we present
the proofs of
Theorems~\ref{teounasemi} and \ref{quartocaso}.

\begin{proof}[Proof of Theorem \ref{teounasemi}] Let $\phi$ be the function
defined by \eqref{solit}, thought of as a half-soliton on $\R^+$ and notice that $\phi(0) = 1$.
Identify $\HH$, the unique half-line of $\G$, with the interval $[0,+\infty)$ and for $\eps>0$ set
\[
u_\eps(x) = \begin{cases} \sqrt \eps \phi (\eps x) & \text{ if  $x\in\HH$ (i.e. if $x \in [0,+\infty)$)} \\
\sqrt \eps & \text{ elsewhere on } \G. \end{cases}
\]
Observe that $\G\setminus\HH\ne\emptyset$, because $\G$ has no terminal point.
Since $\G$ is connected and
$\HH$ is attached to $\G\setminus\HH$ at $x=0$,
we see that $u_\eps \in  H^1 (\G)$ and hence
\[
K_\G \ge\frac {\| u_\eps \|_{L^6 (\G)}^6}{\| u_ \eps \|_{L^2 (\G)}^4 \| u_ \eps' \|_{L^2 (\G)}^2 }  =
\frac {\| \phi \|_{L^6 (\R^+)}^6 + \ell\eps ^3 }
{(\| \phi \|_{L^2 (\R^+)}^2 + \ell\eps )^2 \| \phi' \|_{L^2 (\R^+)}^2 },
\]
where  $\ell$ is the measure (i.e. the total length) of $\G\setminus \HH$.
Since the last quotient tends to $K_{\R^+}$ as $\eps\to 0$,
we have  $K_\G \ge K_{\R^+}$ and, in fact, equality occurs, by \eqref{interK}.
Thus $\muG=\muRp$.

Now the fact that $\elevel_\G(\mu)$ is attained when $\mu \in (\muRp,\muR]$ is the content of Proposition \ref{mainprop}.

Finally, since  $\muRp = \muG$, $\elevel_\G(\muRp)=0$. If $u\in H_{\muRp}^1(\G)$ is such that $E(u,\G) = 0$, then its decreasing rearrangement $u^*$ on $\R^+$ satisfies $E(u^*,\R^+) \le 0$. Then, exactly as in the last part of the proof of Theorem \ref{teobaffo},
this implies that $\G$ is (isometric to) $\R^+$. But this is impossible, since $\G$ has no terminal point.
\end{proof}
To complete the proof of Theorem \ref{quartocaso} we need the following lemma.

\begin{lemma}[Bridge doubling] 
\label{bridge}
Assume $\G$ is noncompact,
 and let $\B$ denote the union of all those edges
that do not belong to any cycle. Then, for  every $u\in H_\mu^1(\G)$,
\begin{equation}
\label{brdoub}
\int_\G |u|^6\,dx +
3\int_\B |u|^6\,dx\leq
3\left(\frac {\mu+3\mu_\B}{\muR}\right)^2
\int_\G |u'|^2\,dx,
\end{equation}
where
\begin{equation}
\label{defmuB}
\mu_\B:= \int_\B |u|^2\,dx.
\end{equation}
\end{lemma}
\begin{proof}
If $\B=\emptyset$, then $\G$ admits a cycle covering and, by Theorem~\ref{teoH},
$\muG=\muR$: then $\mu_\B=0$ and  \eqref{brdoub} reduces to the
Gagliardo-Nirenberg inequality. So we may assume $\B\not=\emptyset$:
the idea behind the proof is that \eqref{brdoub} is still the Gagliardo-Nirenberg inequality,
but for a modified function $\widetilde u$ on a modified graph $\widetilde{\G}$
where we force a cycle covering.

We construct $\widetilde{\G}$ from $\G$, as follows:
for every  edge  $e\in\B$,
we stretch $e$ by a factor $2$, and
we \emph{duplicate} the resulting edge, so that, if the original $e$ had length $\ell$,
there are now \emph{two} edges, each of length $2\ell$, joining
the same two vertices (or emanating from the same vertex, if $e$ is a half-line
and $\ell=+\infty$).
These
two new edges now form a cycle, so that the resulting graph $\widetilde{\G}$
admits a cycle covering.

Given $u\in H^1(\G)$, we construct $\widetilde{u}\in H^1(\widetilde{\G})$ as follows.
First, we let $\widetilde u\equiv u$ on $\G\setminus\B$. Then, for every
edge $e\in\B$,
we duplicate
$u$ (stretched horizontally by a factor $2$) on each of the two copies
of the stretched edge $2 e$ of $\widetilde\G$. Choosing a coordinate $x\in [0,\ell]$ on every $e\in\B$ of length  $\ell$,
this amounts to replacing $u(x)$ over $[0,\ell]$ with \emph{two copies}
of $u(x/2)$, over $[0,2 \ell]$ (these intervals are to be replaced with $[0,+\infty)$
when $e$ is a half-line).
It is then clear that for every $p\ge 1$,
\begin{equation}
\label{eq58}
\int_{\widetilde \G} \left|\widetilde{u}\right|^p\,dx=
\int_{\G\setminus\B} \left|u\right|^p\,dx
+
4 \int_{\B} \left|u\right|^p\,dx = \int_{\G} \left|u\right|^p\,dx
+
3 \int_{\B} \left|u\right|^p\,dx\
\end{equation}
(in particular, when $p=2$ and when $p=6$), and, similarly,
\begin{equation}
\label{eq59}
\int_{\widetilde \G} \left|\widetilde{u}'\right|^2\,dx =
\int_{\G\setminus\B} \left|u'\right|^2\,dx + \int_{\B} \left|u'\right|^2\,dx
= \int_{\G} \left|u'\right|^2\,dx.
\end{equation}
On the other hand, since $\widetilde{\G}$ is covered by cycles,
Theorem~\ref{teoH} gives $\mu_{\widetilde \G}=\muR$, so that
$\widetilde u$ is subject to a
Gagliardo-Nirenberg inequality that can be written
\[
\int_{\widetilde \G} \left|\widetilde{u}\right|^6\,dx\leq
3\left(\frac {\int_{\widetilde \G} \left|\widetilde{u}\right|^2\,dx}{\muR}\right)^2
\int_{\widetilde \G} \left|\widetilde{u}'\right|^2\,dx,
\]
and \eqref{brdoub} follows immediately using \eqref{defmuB}, \eqref{eq58} and \eqref{eq59}.
\end{proof}

\begin{proof}[Proof of Theorem \ref{quartocaso}] We recall that $\muG<\muR$ by assumption.
When $\mu \in (\muG,\muR]$, the fact that $\elevel_\G(\mu)$ is attained follows
from Proposition \ref{mainprop}.

On the other hand, the proof of Proposition \ref{mainprop} cannot be
adapted to the case where
 $\mu = \muG$,
because it is based on the inequality $\elevel_\G(\mu) <0$, that now is
replaced by $\elevel_\G(\muG) =0$. This is not a weakness of the proof:
now, indeed,  any sequence $u_n\in H^1_\mu(\G)$ such that $\Vert u_n'\Vert_{L^2}\to 0$
is, by virtue of \eqref{GN}, a minimizing sequence, so that minimizing sequences
are in general noncompact. To get compactness, one has to carefully select the
minimizing sequence, as follows.

Let $u_n \in H_{\mu_\G}^1(\G)$ be a maximizing sequence for the Gagliardo--Nirenberg inequality,
namely a sequence such that
\begin{equation}
\label{uptoK}
\frac{\|u_n\|_{L^6(\G)}^6}{\|u_n'\|_{L^2(\G)}^2}\; 
\to\; K_\G \mu_\G^2 = 3
\end{equation}
as $n\to \infty$.
Applying inequality \eqref{sharpineq} to $u_n$ keeping in mind that $\mu = \muG$ and $\exM\le \mu \le C$ we can
get rid of $\exM$ and write (for some other $C$)
\begin{equation}
\label{newstima}
\|u_n\|_{L^6(\G)}^6   \le  3\frac{\muG^2}{\muR^2}  \|u_n'\|_{L^2(\G)}^2  +C.
\end{equation}
We first notice that $\|u_n'\|_{L^2(\G)}$ (and hence also $\|u_n\|_{L^6(\G)}^6$) must be bounded. Indeed, if this is not the case,
along a subsequence,
\[
\lim_n \frac{\|u_n\|_{L^6(\G)}^6}{\|u_n'\|_{L^2(\G)}^2} \le \lim_n \left(3\frac{\muG^2}{\muR^2} +\frac{C}{\|u_n'\|_{L^2(\G)}^2}\right) = 3\frac{\muG^2}{\muR^2} < 3,
\]
by \eqref{newstima},
contradicting \eqref{uptoK}.

Now since $\G$ has at least two half-lines, by \eqref{brdoub} we have

\[
\int_\G |u_n|^6\,dx + 3\int_\B |u_n|^6\,dx\leq
3\left(\frac {\mu_\G+3\mu_n}{\mu_\R}\right)^2
\int_\G |u_n'|^2\,dx, \quad \mu_n:= \int_\B |u_n|^2\,dx,
\]
where $\B$ is defined as in Lemma \ref{bridge}.
This allows us to show that $\|u_n\|_{L^6(\G)}$ and $\|u_n'\|_{L^2(\G)}$ are bounded away from zero. If (for some subsequence) the two norms
tend to zero, then clearly also $\mu_n \to 0$. Dividing the preceding inequality by $\|u_n'\|_{L^2(\G)}$ we obtain
\[
\frac{\|u_n\|_{L^6(\G)}^6}{\|u_n'\|_{L^2(\G)}^2} \le 3\left(\frac {\mu_\G+3\mu_n}{\mu_\R}\right)^2 = 3\frac{\muG^2}{\muR^2} +o(1),
\]
and then
\[
\lim_n \frac{\|u_n\|_{L^6(\G)}^6}{\|u_n'\|_{L^2(\G)}^2} <3,
\]
contradicting again \eqref{uptoK}.

Finally, since $u_n$ is bounded in $H^1(\G)$, writing
\[
6E(u_n,\G) = 3 \|u_n'\|_{L^2(\G)}^2 - \|u_n\|_{L^6(\G)}^6 = \|u_n'\|_{L^2(\G)}^2\left( 3 -  \frac{\|u_n\|_{L^6(\G)}^6}{\|u_n'\|_{L^2(\G)}^2}\right)
\]
we see by \eqref{uptoK} that $E(u_n,\G) \to 0 = \elevel_\G(\mu_\G)$. We have thus constructed a minimizing sequence for $E(\,\cdot\,,\G)$ which is bounded and uniformly away from zero.

We may then assume that $u_n\rightharpoonup u$ in $H^1(\G)$, as well as the other usual convergence properties.
If $u\equiv 0$, by Lemma \ref{weakzero},
\[
E(u_n, \G) \ge \frac12 \left(1- \frac{\muG^2}{\muR^2}\right) \|u_n'\|_{L^2(\G)}^2 + o(1).
\]
Since $E(u_n,\G) \to 0$, this  contradicts the construction of the sequence $u_n$. Therefore the weak limit $u$ does not vanish identically, and by Lemma \ref{limit}, it is the required minimizer.
\end{proof}

\end{document}